\documentclass[page-classic]{epl2} 
\usepackage[utf8]{inputenc}  
\usepackage[T1]{fontenc}     
\usepackage[british]{babel}  
\usepackage[scaled=0.86]{berasans}  
\usepackage[colorlinks=true,allcolors=blue]{hyperref}  
\usepackage{lmodern}
\usepackage{multicol}
\usepackage{lipsum}
\usepackage{graphicx} 
\usepackage[babel]{microtype}  
\usepackage{mathtools,cuted}
\usepackage{amsmath,amssymb,amsthm,bm,amsfonts,mathrsfs,bbm} 
\usepackage{setspace}
\usepackage{csquotes}

\usepackage{xspace}  
\usepackage{verbatim}

\newcommand\id{\leavevmode\hbox{\small1\kern-3.3pt\normalsize1}}

\newtheorem{theorem}{Theorem}

\newtheorem{corollary}[theorem]{Corollary}

\title{Bell-CHSH Violation Under Global Unitary Operations: Necessary and Sufficient conditions}

\author{Nirman Ganguly\inst{1,\footnote{ On leave from Department of Mathematics, Heritage Institute of Technology, Kolkata-107, India},\footnote{nirmanganguly@gmail.com}} \and Amit Mukherjee\inst{1}\footnote{amitisiphys@gmail.com} \and Arup Roy\inst{1}\footnote{arup145.roy@gmail.com} \and Some Sankar Bhattacharya\inst{1}\footnote{somesankar@gmail.com} \and Biswajit Paul\inst{2}\footnote{biswajitpaul4@gmail.com} \and Kaushiki Mukherjee\inst{3}\footnote{kaushiki$ \_ $mukherjee@rediffmail.com}}

\shortauthor{Nirman Ganguly \etal}
\institute{                    
  \inst{1} Physics and Applied Mathematics Unit, Indian Statistical Institute, 203 B. T. Road, Kolkata-108, India\\
  \inst{2} Department of Mathematics, South Malda College, Malda, West Bengal, India\\
  \inst{3} Department of Mathematics, Government Girls’ General Degree College, Ekbalpore, Kolkata, India
  }

\abstract{
The relation between Bell-CHSH violation and factorization of Hilbert space is considered here. That is, a state which is local in the sense of the Bell-CHSH inequality under a certain factorization of the underlying Hilbert space can be Bell-CHSH non-local under a different factorization. While this question has been addressed with respect to separability , the relation of the factorization with Bell-CHSH violation has remained hitherto unexplored. We find here, that there is a set containing density matrices which do not exhibit Bell-CHSH violation under any factorization of the Hilbert space brought about by global unitary operations. Using the Cartan decomposition of $ SU(4) $,we characterize the set in terms of a necessary and sufficient criterion based on the spectrum of density matrices. Sufficient conditions are obtained to characterize such density matrices based on their bloch representations. For some classes of density matrices, necessary and sufficient conditions are derived in terms of bloch parameters.  Furthermore, an estimation of the volume of such density matrices is achieved in terms of purity. The criterion is applied to some well-known class of states in two qubits.Since, both local filtering and global unitary operations influence Bell-CHSH violation of a state, a comparative study is made between the two operations. The inequivalence of the two operations(in terms of increasing Bell-CHSH violation) is exemplified through their action on some classes of states.}

\begin{document}

\maketitle
\section{Introduction} Entanglement and non-locality are two inequivalent yet very distinctive features of quantum mechanics\cite{ReviewEnt,ReviewNL,Bell64,CHSH69}. They play a ubiquitous role in many quantum computation and information processing tasks\cite{Bennett93,Bennett92,random, key,dw,game}. Entanglement, being a characteristic trait of a density matix\cite{entbloch} has been characterized in terms of state parameters and bloch representation \cite{entbloch}. Whenever we say that a state is entangled we assume a particular factorization of the underlying Hilbert space. Therefore, a state which is entangled in a certain basis can be separable in a different basis. In fact,for any entangled state(pure and mixed) there exists at least one basis in which it is separable \cite{entglobal}.\\ 
\indent However, there exist separable states which preserve separability under any change of basis, subsequently termed as absolutely separable states \cite{absepzyck}.Equivalently, one may also note that absolutely separable states are exactly those states which preserve separability under the action of any global unitary, i.e., a state $ \sigma_{as} $ will be termed as an absolutely separable state if $ U \sigma_{as} U^{\dagger}$ is separable for any unitary $ U $. They are characterized in terms of their spectrum \cite{absepver,johnston} and also termed as states which are separable from spectrum in view of the open problem pointed out in \cite{Knill}. The set containing absolutely separable states has the same existence in the set of separable states as the separable states have in the space of all density matrices. They form a convex and compact subset of the separable set thus allowing for the construction of witness operators to detect non-absolutely separable states from which useful entanglement can be created under global unitary \cite{witabs}.\\
\indent An equivalent study in the non-local scenario, in terms of violation of the Bell-CHSH inequality\cite{Bell64,CHSH69} was first proposed by some of us\cite{bellGU}, where we showed that there are states which preserve their "\textit{local}" character under any global unitary operation. We termed such states as \textit{absolutely Bell-CHSH local} states and showed that they form a convex and compact subset of the \textit{Bell-CHSH local} set.Precisely, a state $ \sigma_{al} $ is termed as \textit{absolutely Bell-CHSH local} if $ U \sigma_{al} U^{\dagger}  $ is \textit{Bell-CHSH local} for any unitary $ U $.  On a different perspective the utility to detect states which can become Bell-CHSH non-local under global unitary was also highlighted in\cite{bellGU} and a subsequent mechanism was suggested to detect such states.\\
\indent However, given a state , how does one identify that it is \textit{absolutely Bell-CHSH local} and what is the size of the set containing \textit{absolutely Bell-CHSH local} states? The answer to these two questions form the main constituents of the present work. We derive necessary and sufficient conditions based on the spectrum of the density matrix to identify whether it is \textit{absolutely Bell-CHSH local}. In the context of quantum state tomography \cite{}, it is important to identify the number of state parameters required to acquire knowledge about some property of the state. Therefore,we also derive sufficient conditions based on the bloch parameters of a density matrix for the above-mentioned identification.We find that, based on the bloch parameters necessary and sufficient conditions can be derived to identify some special class of density matrices.  Next, we find a characterization of the set in terms of purity of density matrices which alternatively provides an estimation of the size of the \textit{absolutely Bell-CHSH local} set. Specifically, we find a ball of a certain radius within which all states are \textit{absolutely Bell-CHSH local}. Further our analysis is validated by illustrations from well known class of states.\\
\indent Local filtering operations \cite{verstraetebell} can have a non-trivial effect on the nonlocality of a state just like global unitary action. However, no local filtering operation can increase the Bell-CHSH violation of Bell diagonal states \cite{verstraetebell}. Although very different in paradigm we have found an interesting similarity between these two operations as global unitary action on Bell diagonal states also fails to increase CHSH violation. However the inequivalence between these two operations come to the fore when we see their action on other classes of states, which we have noted in our work.\\
\indent Our work is organized in the following way. The next two sections characterize the \textit{absolutely Bell-CHSH local} set in terms of eigenvalues and the bloch parameters of a density matrix. The section that follows evaluates the size of the \textit{absolutely Bell-CHSH local} set in terms of purity and derives the \textit{absolutely Bell-CHSH local} ball. Illustrations from well known classes of states follows next followed by a comparison between global unitary and local filtering operations. We end with conclusions and comments on possible courses of future work.
\section{Characterization of Bell-CHSH local states based on spectrum}
A density matrix living in $ \mathbf{B}(\textbf{C}^2 \otimes \textbf{C}^2) $ can be expressed in the Hilbert-Schmidt basis as ,
\begin{equation}
\sigma = \frac{1}{4}[I \otimes I + \vec{u}.\vec{s} \otimes I + I \otimes \vec{v}.\vec{s} + \Sigma_{i,j=1}^{3} t_{ij} s_i \otimes s_{j} ]
\end{equation}
Here, $ \vec{u} , \vec{v}  $ are the local Bloch vectors and $ t_{ij} = Tr[\sigma(s_i \otimes s_j)] $, $ s_{i}  $ are the Pauli matrices.
In a much celebrated work \cite{horobell}, the Horodecki family derived a necessary and sufficient condition to verify whether a state is local with respect to the Bell-CHSH inequality, based on the parameters of the state. The condition was based on the value of a function $ M(\sigma_{l}) = \lambda_{max1} + \lambda_{max2}; \lambda_{max1}, \lambda_{max2}$ being the maximum two eigen values of $ Y = T^{\dagger}T $ where $ T = [t_{ij}]$ is the correlation matrix of $ \sigma_{l} $ . It was then stated that a state $ \sigma_{l} $ is Bell-CHSH local iff $ M(\sigma_{l}) \le 1 $. The maximal Bell-CHSH violation for any state $ \chi $ was then showed to be $ 2 \sqrt{M(\chi)} $. In this light, the \textit{absolutely Bell-CHSH local} set can also be defined as $ \mathbf{AL} = \lbrace \sigma_{al} : M(U \sigma_{al} U^{\dagger}) \le 1 , \forall U \rbrace $.\\
\indent The first main result of our work , now follows below:
\begin{theorem}
	A state $\sigma$ is absolutely Bell-CHSH local if and only if $(2a_1+2a_2-1)^2 + (2a_1+2a_3-1)^2 \le 1 $, where $ a_1, a_2, a_3 $ are the highest three eigenvalues of $ \sigma $ in a decreasing order.  
\end{theorem}  
\begin{proof}
	Define a function $ F(\sigma)= \underset{U \in SU(4)}{Max} M(U \sigma U^\dagger)$. It is easy to note that, $ F(\sigma) \le 1$, iff $ \sigma \in \mathbf{AL}$. In view of the Cartan decomposition of $ SU(4) $ \cite{optcreate}, any $ U \in SU(4) $ can be decomposed as ,
	\begin{equation}
	U = U_A \otimes U_B U_d V_A \otimes V_B \label{cartan1}
	\end{equation}
	where $ U_A , U_B , V_A , V_B$ are local unitaries and $ U_d $ is the basic non-local unitary. A local unitary $ U_L $ changes the correlation matrix $ T_\sigma $ corresponding to $ \sigma $ in the following way, $ T^ \prime = Q_1 T_\sigma Q_2^{\dagger}$, where $ Q_i s$ are rotation matrices with $ det(Q_i)=1 , Q_i^{\dagger} Q_i=I$ \cite{infhoro}. Therefore, 
	\begin{eqnarray}
	T^ {\prime^{\dagger}} T^ \prime = Q_2 T_\sigma^{\dagger} T_\sigma Q_2
^{\dagger}	\end{eqnarray}
Since, the above relation signifies a similarity transformation, the eigenvalues of $ T_\sigma^{\dagger} T_\sigma  $ remain unchanged signalling the invariance of $ M(\sigma) $ under local unitary transformation. Therefore, $ M(\sigma) $ can only be maximized by the action of the basic non-local operator $ U_{d} $. As a result, in order to check whether a state is \textit{absolutely Bell-CHSH local}, it is enough to see the change under the action of the basic non-local unitary. The operator $ U_d $ is diagonal in the magic basis and hence will also be diagonal in the Bell basis\cite{optcreate}, and thus can be expressed as ,
\begin{equation}
U_d = \Sigma_{k=1}^{4} e^{-\textbf{i}\lambda_k} \vert \phi_k \rangle \langle \phi_k \vert
\end{equation}
with $ \textbf{i} = \sqrt{-1}$. Here, $ \vert \phi_k \rangle  $ are the maximally entangled Bell states and $ \lambda_k s$ are given by ,
$ \lambda_1 = x -y +z , \lambda_2 = -x +y +z , \lambda_3 = -x-y-z , \lambda_4 = x +y -z$. ($ x,y,z \in [0,2\pi]$ ).\\
The proof will be done considering three different cases underlined as below:\\
\textbf{Case-I} Consider the Bell diagonal state given by :
\begin{equation}
\sigma_{bell} = a_1 \vert \phi_1 \rangle \langle \phi_1 \vert + a_2 \vert \phi_2 \rangle \langle \phi_2 \vert + a_3 \vert \phi_3 \rangle \langle \phi_3 \vert +a_4 \vert \phi_4 \rangle \langle \phi_4 \vert
\end{equation}
Without any loss of generality we assume $ a_1 \ge a_2 \ge a_3 \ge a_4 $. Since, the basic non-local operator $ U_d $ is itself diagonal with respect to the Bell basis and the factor $ e^{-\textbf{i}\lambda_k} $ cancels out, it is easy to see that $ \sigma_{bell} $ remains invariant under the transformation $ U_d \sigma_{bell} U_d^{\dagger}$. This when clubbed with the fact that local unitaries do not change the eigenvalues of $ T^{\dagger}_{\sigma_{bell}} T_{\sigma_{bell}} $, imply that even after a global unitary action the eigenvalues of $ T^{\dagger}_{\sigma_{bell}} T_{\sigma_{bell}} $ remain unchanged($ T_{\sigma_{bell}} $ is the correlation matrix corresponding to $ \sigma_{bell} $). Thus, $ M(\sigma_{bell}) $ cannot be increased further by global unitary action. \\
For the Bell diagonal state , $ M(\sigma_{bell}) = (2a_1+2a_2-1)^2 + (2a_1+2a_3-1)^2 $ (as $ a_4 = 1- a_1-a_2-a_3 $). Now, $(2a_1+2a_2-1)^2 + (2a_1+2a_3-1)^2 \le 1 $ iff the Bell diagonal state is \textit{Bell-CHSH local}. In view of the fact that global unitary does not maximize  $ M(\sigma_{bell}) $ , we can conclude that $(2a_1+2a_2-1)^2 + (2a_1+2a_3-1)^2 \le 1 $ iff $\sigma_{bell}$ is \textit{absolutely Bell-CHSH local}. In fact, this also shows that whenever a Bell diagonal state is \textit{Bell-CHSH local} then it is \textit{absolutely Bell-CHSH local}. This is a typical feature of the Bell diagonal states.\\
\textbf{Case-II} Consider now, a state $ \sigma_{comp} $ which is diagonal in the computational basis,
\begin{equation}
\sigma_{comp}= a_1 \vert 00 \rangle \langle 00 \vert + a_2 \vert 01 \rangle \langle 01 \vert + a_3 \vert 10 \rangle \langle 10 \vert + a_4 \vert 11 \rangle \langle 11 \vert 
\end{equation}
Again, without loss of generality it is assumed that $ a_1 \ge a_2 \ge a_3 \ge a_4 $. After the action of $ U_{d} $, the eigenvalues of $ T^{\dagger}_{ \sigma_{comp}} T_{\sigma_{comp}} $ change and are found to be as given below:

\begin{eqnarray}
&&A= [sin[2x-2y](2a_1+a_2+a_3-1)+sin[2x+2y](a_2-a_3)]^2 \nonumber \\
&&B= [sin[2x-2y](2a_1+a_2+a_3-1)\nonumber-sin[2x+2y](a_2-a_3)]^2 \nonumber \\
&&C= [2a_2+2a_3-1]^2 
\end{eqnarray}

Therefore, $ F(\sigma_{comp})  = Max (\underset{x,y}{Max} (A+B), \underset{x,y}{Max} (A+C),\underset{x,y}{Max} (B+C))$.
After some algebraic calculation it is found that :
\begin{eqnarray}
\underset{x,y}{Max} (A+B) = (2a_1+2a_2-1)^2 + (2a_1+2a_3-1)^2 \nonumber\\
\underset{x,y}{Max} (A+C) = (2a_1+2a_2-1)^2 + (2a_2+2a_3-1)^2 \nonumber\\
\underset{x,y}{Max} (B+C) = (2a_1+2a_2-1)^2 + (2a_2+2a_3-1)^2 \nonumber\\
\end{eqnarray}
Therefore $ F(\sigma_{comp})  = \underset{x,y}{Max} (A+B) = (2a_1+2a_2-1)^2 + (2a_1+2a_3-1)^2 $.\\
 Hence $ \sigma_{comp} \in \mathbf{AL} $ iff  $ (2a_1+2a_2-1)^2 + (2a_1+2a_3-1)^2 \le 1 $.\\
\textbf{Case-III} Now, finally consider a state $ \sigma_{arb} $ written in any arbitrary basis.\\
For any given spectrum the maximal Bell-CHSH violation is obtained at the respective Bell-diagonal state\cite{verstraetebell}, where the quantity $ M(\sigma_{bell}) = (2a_1+2a_2-1)^2 + (2a_1+2a_3-1)^2  $, as already obtained in Case-I. The same maximal value is obtained by maximizing the corresponding value for states diagonal in the computational basis in \textbf{Case-II}. Now,
one can always find a unitary which changes the $ \sigma_{arb} $ to a state diagonal in the computational basis and repeat the same steps to obtain the maximal value.\\
Hence,$ F(\sigma_{arb})  =(2a_1+2a_2-1)^2 + (2a_1+2a_3-1)^2 $ and thus  $ \sigma_{arb} \in \mathbf{AL} $ iff  $ (2a_1+2a_2-1)^2 + (2a_1+2a_3-1)^2 \le 1 $.\\
Thus, combining the above cases, a state $\sigma$ is \textit{absolutely Bell-CHSH local} if and only if $(2a_1+2a_2-1)^2 + (2a_1+2a_3-1)^2 \le 1 $. Hence, the proof.
\end{proof}
One may note the following corollary , in view of the theorem,
\begin{corollary}
The reduced state $ \sigma_{AB} $ of any pure three-qubit state $ \vert \Upsilon \rangle_{ABC} $ is absolutely Bell-CHSH local if and only if $ \sigma_{C} $ is the maximally-mixed state.
\end{corollary} 
\begin{proof}
Consider that $ c= \vert \overrightarrow {\textbf{c}}  \vert$ ,where $ \overrightarrow {\textbf{c}} $ is the bloch vector for $ \sigma_{C} $. Then the eigenvalues of $ \sigma_{AB} $ are $ \lbrace (1+c)/2 , (1-c)/2 , 0 , 0 \rbrace $. Hence, in view of the above theorem , $ \sigma_{AB} $ is \textit{absolutely Bell-CHSH local} iff $ c $ vanishes.
\end{proof}
\section{Characterization of Absolutely Bell-CHSH local states based on bloch parameters }
Quantum state tomography identifies the number of ideal measurements to be done on a state to reveal some property of the state \cite{tomo}. Therefore, it becomes very important to have a characterization in terms of bloch parameters of a state.\\
As every density matrix has a Hilbert-Schmidt decomposition, it is pertinent whether one can comment on the \textit{absolute Bell-CHSH local} character of a density matrix, based on the parameters in the decomposition.
Recall that, a density matrix living in $ \mathbf{B}(\textbf{C}^2 \otimes \textbf{C}^2) $ can be expressed in the Hilbert-Schmidt basis as ,
\begin{equation}
\sigma = \frac{1}{4}[I \otimes I + \vec{u}.\vec{s} \otimes I + I \otimes \vec{v}.\vec{s} + \Sigma_{i,j=1}^{3} t_{ij} s_i \otimes s_{j} ]
\end{equation}
Here, $ \vec{u} , \vec{v}  $ are the local Bloch vectors and $ t_{ij} = Tr[\sigma(s_i \otimes s_j)] $, $ s_{i}  $ are the Pauli matrices. In what follows below, we see that in some cases we can arrive at necessary and sufficient conditions to guarantee that a density matrix is indeed \textit{absolutely Bell-CHSH local} . 
\begin{theorem}
	A Bell diagonal state is \textit{absolutely Bell-CHSH local} iff $Max(t_{11}^2+ t_{22}^2 ,t_{11}^2+ t_{33}^2 , t_{22}^2+ t_{33}^2 ) \le 1 $, where $ t_{ii} $ are the diagonal elements of the correlation matrix of the Bell diagonal state.
\end{theorem}
\begin{proof}
	The proof follows quite easily in view of the fact that the basic non-local unitary operator does not change $ \sigma_{bell} $. Further, the local bloch vectors are zero for a Bell diagonal state and its correlation matrix is diagonal. Hence the desired eigenvalues are $ t_{11}^2, t_{22}^2 ,t_{33}^2 $ . Now, $ M(\sigma_{bell}) = Max(t_{11}^2+ t_{22}^2 ,t_{11}^2+ t_{33}^2 , t_{22}^2+ t_{33}^2 ) $. Since, eigenvalues of $ T^{\dagger}_{\sigma_{bell}} T_{\sigma_{bell}} $ do not change even with a global unitary, $ M(\sigma_{bell})$ does not change. Therefore, $Max(t_{11}^2+ t_{22}^2 ,t_{11}^2+ t_{33}^2 , t_{22}^2+ t_{33}^2 ) \le 1 $ iff $ \sigma_{bell} $ is \textit{absolutely Bell-CHSH local}. 
\end{proof}
In order to find equivalent conditions for other states, we need to note that unlike Bell diagonal states, the eigenvalues of $ T^{\dagger} T$ corresponding to other states will change due to the action of a global unitary. \\
In order to note the change in terms of the bloch parameters we use an alternative decomposition of an unitary operator as given in \cite{entbloch},
\begin{equation} 
U = (U_1 \otimes U_2) U_d(\theta_1,\theta_2,\theta_3) (U_3 \otimes U_4) \label{cartan2}
\end{equation}
where again $ U_{k} $ are local unitaries and $ U_d = exp[\frac{\textbf{i}}{2}(\theta_1 s_1 \otimes s_1 + \theta_2 s_2 \otimes s_2 + \theta_3 s_3 \otimes s_3)] $ is the basic non-local unitary. We again note that it is the basic non-local operator which is responsible for any change in the eigenvalues of $ T^{\dagger} T$. On action of the basic non-local unitary operator, the modifications in $\vec{u},\vec{v}, T$ are given below \cite{entbloch},
\begin{eqnarray}
u_k^\prime = u_k cos \theta_i cos \theta_j + v_k sin \theta_i sin \theta_j + \epsilon_{ijk}(T_{ij} cos \theta_i sin \theta_j - T_{ji} sin \theta_i cos \theta_j) \\
v_k^\prime = v_k cos \theta_i cos \theta_j + u_k sin \theta_i sin \theta_j + \epsilon_{ijk}(T_{ji} cos \theta_i sin \theta_j - T_{ij} sin \theta_i cos \theta_j) \\
T_{ij}^\prime = T_{ij}cos \theta_i cos \theta_j + T_{ji} sin \theta_i sin \theta_j - \epsilon_{ijk}(u_k cos \theta_i sin \theta_j - v_k sin \theta_i cos \theta_j)
\end{eqnarray}
The indices $i,j,k$ are distinct in the first two equations and $\epsilon_{ijk}$ is the Levi-Civita symbol.\\
For states which are diagonal in the computational basis, the only non-zero parameters in its Hilbert-Schmidt representation are $u_3,v_3,t_{33}$. As a result we have the following result,
\begin{theorem}
A state $\sigma_{comp}$ diagonal in the computational basis is absolutely Bell-CHSH local iff
$ Max[\underset{\theta_1,\theta_2,\theta_3}{Max}(t_{33}^2 + c_{11}^2), \underset{\theta_1,\theta_2,\theta_3}{Max}(t_{33}^2 + c_{22}^2) ,\underset{\theta_1,\theta_2,\theta_3}{Max}(c_{11}^2 + c_{22}^2 )] \leq 1 $,
where $ c_{11}=u_3 cos \theta_2 sin \theta_1 - v_3 sin \theta_2 cos \theta_1, c_{22}=v_3 cos \theta_2 sin \theta_1 - u_3 sin \theta_2 cos \theta_1 $. ($\theta_i s$ being the parameters of the basic non-local operator as given in eq. (\ref{cartan2}) )
\end{theorem} 
\begin{proof}
As noted earlier, the basic non-local operator is the one responsible for changing the eigenvalues of the  matrix $T^{\dagger} T$. Now, since the only non-zero parameters in the Hilbert-Schmidt representation of a density matrix diagonal in the computational basis are $u_3 , v_3 , t_{33}$, the eigenvalues of the changed $(T^\prime)^{\dagger} T^\prime$ are $t_{11}^2 , c_{11}^2 , c_{22}^2$. Therefore,
\begin{equation}
 F(\sigma_{comp})= Max[\underset{\theta_1,\theta_2,\theta_3}{Max}(t_{33}^2 + c_{11}^2), \underset{\theta_1,\theta_2,\theta_3}{Max}(t_{33}^2 + c_{22}^2) ,\underset{\theta_1,\theta_2,\theta_3}{Max}(c_{11}^2 + c_{22}^2 )]
 \end{equation}
   Thus , $Max[\underset{\theta_1,\theta_2,\theta_3}{Max}(t_{33}^2 + c_{11}^2), \underset{\theta_1,\theta_2,\theta_3}{Max}(t_{33}^2 + c_{22}^2) ,\underset{\theta_1,\theta_2,\theta_3}{Max}(c_{11}^2 + c_{22}^2 )] \le 1$ iff $\sigma_{comp}$ is \textit{absolutely Bell-CHSH local}.
\end{proof}
For any arbitrary density matrix $\sigma$ it is in general difficult to arrive at the necessary and sufficient conditions in a closed form. The difficulty arises from the difficulty in the computation of the eigenvalues of $(T_{\sigma}^\prime)^{\dagger} T^\prime_{\sigma}$ after the action of the basic non-local unitary. However one may derive a sufficient condition with the following observation. The trace of $(T_{\sigma}^\prime)^{\dagger} T^\prime_{\sigma}$ is necessarily greater than the sum of its two greater eigenvalues. Formally, we can get the sufficient condition as given below,
\begin{theorem}
If for any density matrix $\sigma$, $\underset{\theta_1,\theta_2,\theta_3}{Max} Tr((T_{\sigma}^\prime)^{\dagger} T^\prime_{\sigma}) \le 1$ then $\sigma\in \mathbf{AL}$ 
\end{theorem}
However, it should be noted in order to calculate the trace, individual changes in the bloch parameters has to be taken into account which can be cumbersome. The absence of some of the bloch parameters can simplify the process.
\section{Comparison with purity and Absolutely Bell-CHSH local ball}
\textit{Absolutely Bell-CHSH local} states can also be located in terms of their purity. This in turn allows them to be characterized in terms of their distance from the maximally mixed state.\\
\indent Analogous to a similar concept in the separability problem \cite{gurvits},it is useful to find an estimation of the size of the set containing \textit{absolutely Bell-CHSH local} states and also to find the maximum possible purity beyond which there cannot be any \textit{absolutely Bell-CHSH local} state. We pose the following optimization problems to answer the questions. \\
\indent Let us assume that the eigenvalues of the density matrix are $ a_1 , a_2 , a_3 ,a_4  $ in descending order. The condition for a state to be \textit{absolutely Bell-CHSH local} can also be re-framed as $(a_1-a_4)^2+(a_2 - a_3)^2 \le 1/2$. The problem of finding the maximum purity beyond which there cannot be any \textit{absolutely Bell-CHSH local} state can now be posed as :
\begin{eqnarray}
&&Maximize (a_1^2 + a_2^2 + a_3^2 + a_4^2),\nonumber \\
&&subject ~~ to : \nonumber \\
&&(i) (a_1-a_4)^2+(a_2 - a_3)^2 \le 1/2 \nonumber \\
&&(ii) a_1+a_2+a_3+a_4 =1 \nonumber \\
&&(iii) 1 \ge a_1 \ge a_2 \ge a_3 \ge a_4 \ge 0
\end{eqnarray} 
\indent A numerical computation yields the answer as  $ \frac{5}{8} $. The solution signifies that if the purity of a density matrix is greater than $ \frac{5}{8} $, it cannot be a \textit{absolutely Bell-CHSH local} state.\\
\indent The converse question of finding the minimum purity below which there cannot be any \textit{non-absolutely Bell-CHSH local} state can be posed as :
\begin{eqnarray}
&&Minimize (a_1^2 + a_2^2 + a_3^2 + a_4^2),\nonumber \\
&&subject ~~ to : \nonumber \\
&&(i) (a_1-a_4)^2+(a_2 - a_3)^2 > 1/2 \nonumber \\
&&(ii) a_1+a_2+a_3+a_4 =1 \nonumber \\
&&(iii) 1 \ge a_1 \ge a_2 \ge a_3 \ge a_4 \ge 0
\end{eqnarray}  
\indent Numerically, the answer is obtained as $ \frac{1}{2} $. This signifies that all states with a purity $ \le \frac{1}{2} $ is \textit{absolutely Bell-CHSH local}.\\
The above result can also be framed in terms of the Frobenius norm $\Vert A \Vert^2 = Tr(A^\dagger A)$. Specifically, $Tr(\rho^2) \le 1/2 \implies \Vert \rho - I/4 \Vert \le 1/2$, which is the \textit{absolutely Bell-CHSH local} ball. Every state within this ball is \textit{absolutely Bell-CHSH local}.  ($I/4$ is the maximally mixed state).\\
\indent On the other hand $Tr(\rho^2) > 5/8 \implies \Vert \rho - I/4 \Vert >  \sqrt{3}/2 \sqrt{2}$. Thus any state outside this ball cannot be 
\textit{absolutely Bell-CHSH local}.
\section{Illustrations} We now provide some illustrations on the application of our criterion. \\
(i) \textit{Absolutely Separable states -} Any absolutely separable state will be \textit{absolutely Bell-CHSH local}. This is because , absolutely separable states preserve their separability under global unitary operation. In set theoretic language, if we denote the set containing the absolutely separable states by $\mathbf{AS}$, then $\mathbf{AS}$ forms a subset of $\mathbf{AL}$. \\
(ii) Any pure product state cannot be \textit{absolutely Bell-CHSH local}, as they can always be converted to a pure entangled state by some global unitary \cite{ent global}.\\
(iii) If we take a state diagonal in computational basis, say $ \sigma_{mix} = \frac{1}{2}\vert 00 \rangle \langle 00 \vert + \frac{1}{2} \vert 11 \rangle \langle 11 \vert  $, then the state is \textit{absolutely Bell-CHSH local} as, $ (1/2-0)^2 + (1/2-0)^2 \le 1/2 $.\\
One can verify this from the perspective of bloch parameters. The only bloch parameter in this case is $ t_{33} = 1 $. Using the result obtained for states diagonal in computational basis , one sees that the state $ \in \mathbf{AL} $.\\
(iv) \textit{Werner state-} The Werner state is given as \cite{wernerstate}, $ \sigma_{wer} = p \vert \psi^{-} \rangle \langle \psi^{-} \vert + \frac{1-p}{4} I$($ \vert \psi^{-} \rangle = \frac{\vert 01 \rangle - \vert 10 \rangle }{\sqrt{2}} $). The state is absolutely separable for $ p \le 1/3 $, hence \textit{absolutely Bell-CHSH local} there. The eigenvalues are $ \lbrace (1+3p)/4 , (1-p)/4 , (1-p)/4 , (1-p)/4 \rbrace $. For $ p \le 1/\sqrt{2} $ it is \textit{absolutely Bell-CHSH local}. One may note that, the Werner state is a Bell diagonal state and it is \textit{Bell-CHSH local} upto $ 1/\sqrt{2} $.\\
This characterization can also be done in terms of the bloch parameters. As for the Werner state $ t_{11}=t_{22}=t_{33} = -p$, $ Max(t_{11}^2 + t_{22}^2, t_{11}^2 + t_{33}^2, t_{33}^2 + t_{22}^2) = 2p^2$. Therefore, for $ p \le 1/\sqrt{2} $, it is \textit{absolutely Bell-CHSH local}.Once again this reiterates the fact whenever a Bell diagonal state is \textit{Bell-CHSH local} it is \textit{absolutely Bell-CHSH local}.\\
 Since, the Werner states are entangled for $ p > 1/3 $, this also indicates that the absolutely separable states form a proper subset of the absolutely Bell-CHSH local states.\\
(v) \textit{Gisin states-} The Gisin states were proposed in \cite{gisin}.
Let $ \vert \psi_{\theta} \rangle  = sin \theta \vert 01 \rangle + cos \theta \vert 10 \rangle $ and $ \sigma_{mix} = \frac{1}{2}\vert 00 \rangle \langle 00 \vert + \frac{1}{2} \vert 11 \rangle \langle 11 \vert $ . Then the Gisin state is written as \cite{negentropy} ,
\begin{equation}
\sigma_{G} = \lambda \vert \psi_{\theta} \rangle  \langle \psi_{\theta} \vert + (1- \lambda) \sigma_{mix}
\end{equation}
with $ 0 < \theta < \pi/2, 0 \le \lambda \le 1 $.
We need to deal here taking two cases into consideration ,\\
(a) For $ \lambda \ge 1/3 $, the eigenvalues in the descending order are $ \lbrace \lambda , (1-\lambda)/2 , (1-\lambda)/2 , 0 \rbrace $. Thus, for $ 1/3 \le \lambda \le 1/\sqrt{2} $, the Gisin states are \textit{absolutely Bell-CHSH local}.\\
(b) For $ \lambda < 1/3 $, the eigenvalues in the descending order are the Gisin states are $ \lbrace (1-\lambda)/2 , (1-\lambda)/2 , \lambda , 0 \rbrace $. Therefore, in order to be \textit{absolutely Bell-CHSH local} , one must have $ 5 \lambda^{2} - 4 \lambda \le 0 $. This implies that the Gisin states are  \textit{absolutely Bell-CHSH local} if $ \lambda < 1/3 $. \\
Hence, combining the results above one obtains that the Gisin states are \textit{absolutely Bell-CHSH local} for $ 0 \le \lambda \le  1/\sqrt{2} $.\\
The above characterization can also be done in terms of bloch parameters as given below. For the Gisin states, the non-zero parameters in the bloch representation are $ u_3 , v_3 , t_{11} , t_{22} , t_{33} $. The changed correlation matrix $ T^{\prime} $ due to the action of the basic non-local operator is,
\begin{equation}
\begin{bmatrix}
x11 & x12 & 0\\
x21 & x22 & 0\\
0 & 0 & x33
\end{bmatrix}
\end{equation}
where $ x11 = t_{11} , x22 = t_{22} , x33 = t_{33}, x12 = v_3 cos \theta_2 sin \theta_1 - u_3 sin \theta_2 cos \theta_1 , x21 = u_3 cos \theta_2 sin \theta_1 - v_3 sin \theta_2 cos \theta_1$. On calculation of the eigenvalues of the corresponding symmetric matrix with the explicit values of the parameters and subsequent maximization, it is found that the Gisin states are in $ \mathbf{AL} $ for $ \lambda \le 1/\sqrt{2} $.
\section{Local filtering and Global unitary operations}
Nonlocality of certain local entangled states can be revealed by using local filtering operations before the standard Bell test \cite{gisin,filterpopescu}. This phenomenon was termed as 'hidden nonlocality'. In particular, Gisin \cite{gisin} presented a class of mixed two qubit states that do not violate Bell-CHSH inequality in the standard Bell test but do so when judicious local filtering operations are applied to the state before performing a standard Bell test.  As both global unitary operation and local filtering generate nonlocality so it is important to compare them. Following the work of \cite{gisin}, several aspects of local filtering in Bell test have been discussed \cite{verstraetefilter}. Recently, a necessary and sufficient condition for any two qubit state to remain Bell-CHSH local after application of local filtering operation has been derived in \cite{ghosh13}. In particular, for Bell-diagonal states the above necessary and sufficient condition can be expressed in terms of their spectrum, as follows \cite{ghosh13},\\
\textit{Any Bell-diagonal state $\sigma_{bell}$ remains Bell-CHSH local after the application of local filtering if and only if $(2 a_1 + 2 a_2 -1)^2 + (2 a_1 + 2 a_3 -1)^2 \leq 1$ where $a_1$, $a_2$, $a_3$ are the highest three eigen values of $\sigma_{bell}$.}\\
This result then implies that local filtering operation cannot increase nonlocality of any Bell diagonal state. The fact that local filtering operations cannot increase Bell-CHSH violation of Bell diagonal states was also proven in \cite{verstraetebell}. In fact, in addition to that, the above condition matches our criterion for absolute Bell-CHSH locality. So both global unitary and local filtering operations are unable to increase nonlocality of Bell-diagonal states. But we will see that these two operations are inequivalent for some other class of two qubit states as far as enhancement of nonlocality is concerned.\\
(i) Consider the following two qubit state $\rho_f = q|\psi_{-}\rangle\langle \psi_{-} | + \frac{1-q}{2} (|00\rangle \langle 00| + |01\rangle \langle 01|) $ where $0 \leq q \leq 1$.\\ 
This state exhibits nonlocality for any $q > 0$ when an appropriate local filtering is used \cite{fhirsch}. But according to our criterion, this state is absolutely Bell-CHSH local if and only if $q \leq 0.5673054$. Here local filtering gives advantage over global unitary operation to increase nonlocality of this state in the range $0 < q \leq 0.5673054$. \\ 
(ii) Consider the noisy partial entangled state $\rho_g = p|\psi_{\theta}\rangle\langle \psi_{\theta} | + \frac{(1-p)}{4}I$ where $|\psi_{\theta}\rangle = \cos\theta |01\rangle + \sin \theta |10\rangle$ and $0\leq p \leq 1$. This state does not violate CHSH inequality after the application of local filtering if and only if $q^2(1-\cos 4\theta) > 2\sqrt{1+2q-q^2-2q^2\cos4\theta}$ but violates our criterion when $q > \frac{1}{\sqrt2}$. Hence in this case, there exists a range of state parameters where global unitary operation gives advantage over local filtering operations pertaining to enhancement of nonlocality. 
\section{Conclusion} Bell-CHSH non-locality \cite{Bell64,CHSH69} which is more of a statistical feature exhibited by quantum states, had been first characterized in terms of the state parameters in \cite{horobell}. However, unlike absolute separability the relation between the violation Bell-CHSH inequality and the change of basis of the underlying Hilbert space has remained unexplored. Any change in basis is brought about by global unitary operations. Here, in this work, we have made a probe on quantum states which cannot be made Bell-CHSH non-local even by global unitary operations. This also entails the characterization of quantum states which preserve their Bell-CHSH local character under any factorization of the underlying Hilbert space.\\
\indent Using the Cartan decomposition of $ SU(4) $, we have laid down criterion to  identify such states in terms of their spectrum and their Hilbert-Schmidt representations. The volume of such states have been estimated in terms of purity. The criterion find support from illustrations from different class of states.\\
\indent Both global unitary and local filtering operations fail to increase CHSH violation of Bell diagonal states. However, the two actions are completely inequivalent(in terms of increasing Bell-CHSH violation). A comparison between the two operations have been done in our work through their actions on some quantum states.\\
\indent The present work also raises some pertinent questions. One may attempt to find such conditions for other Bell inequalities in two qubits as well as for higher dimensions. Extension of the concept to include multipartite Bell inequalities is another direction of useful investigation.\\
\acknowledgements
 We acknowledge Prof.Michael Hall for his enriching inputs. We would also like to gratefully acknowledge fruitful discussions with Prof.Guruprasad Kar and Prof. Sibasish Ghosh. We also thank Tamal Guha and Mir Alimuddin for useful discussions. AM acknowledges support from the CSIR project 09/093(0148)/2012-EMR-I.   
\end{document}